\documentclass[preprint,5p,times,twocolumn,authoryear]{elsarticle}

\usepackage{amsmath,amssymb,mathtools}
\usepackage{amsthm}
\usepackage{enumitem}
\usepackage{graphicx}
\usepackage{algorithm}
\usepackage{algpseudocode}
\usepackage{url}
\newtheorem{theorem}{Theorem}

\newcommand{\norm}[1]{\lVert #1 \rVert}

\journal{Cognitive Systems Research}

\begin{document}

\begin{frontmatter}

\title{Interpretation, Learning, and Empathy as One Constraint:\\ A Residual-Adequacy Architecture with Accountable Abstention}

\author{Chainarong Amornbunchornvej}
\address{National Electronics and Computer Technology Center (NECTEC), 112 Phahonyothin Road, Khlong Nueng, Khlong Luang, Pathum Thani 12120, Thailand}
\ead{chainarong.amo@nectec.or.th}

\begin{abstract}
An agent must act on the situation before it, learn what it cannot yet represent, and model other agents well enough to coordinate. These faculties are usually realized by separate mechanisms, yet they share a failure mode: the situation can exceed what the agent can currently represent, and the honest response is then a principled refusal that says what was missing. We develop a small cognitive architecture in which these limits arise from a single quantity. An Interpretation--Decision Unit (IDU) interprets a content vector through a family of regimes --- local representational frames with private bases --- and decides which actions the interpretation licenses; a scalar residual of the content against the active regimes' representational scope drives the unit. Low residual with a clean licensing emits an action, otherwise the unit re-interprets, attempts a description-length-justified expansion of its representation, or halts with a typed, witnessed terminal. We prove the unit is total and deterministic: for any content and fixed configuration it halts in finitely many bounded-cost steps with a unique terminal witness, so abstention carries its cause by construction. We then show, by binding the architecture's open parameters without changing its mechanics, that the same residual-against-scope constraint recovers three independently documented phenomena at three scopes: the typology of not-knowing (typed abstention), a misunderstanding between agents that is forced, localized to one shared concept, and structurally invisible to the agent committing it (bounded empathy), and prerequisite dependence in learning derived from a bounded focus window rather than posited (developmental prerequisites). Each instantiation is worked for both a natural and an artificial agent and states a falsifiable prediction, showing that the same constraint can model limits in both human and machine cognition. The account contributes a unification --- three faculties as one constraint --- and a notion of accountable abstention that is typed and witnessed by construction.
\end{abstract}


\begin{keyword}
cognitive architecture \sep representation \sep abstention \sep theory of mind \sep learnability \sep minimum description length
\end{keyword}

\end{frontmatter}

\section{Introduction}

An agent embedded in the world must do three things that are usually studied apart: act on the situation in front of it, learn what it cannot yet represent, and model other agents well enough to coordinate. Each has a mature literature --- bounded and selective action, developmental and curriculum learning, theory of mind --- and each is typically realized by its own mechanism. Yet the three share a common failure mode that the separate treatments obscure: an agent can fail to act, fail to learn, or fail to understand another \emph{for the same underlying reason} --- the situation exceeds what the agent can currently represent. When that happens, the honest response is not a confident output but a principled refusal to commit, accompanied by an account of what was missing. Most systems cannot give that account: they either force an answer or abstain with a single undifferentiated signal of low confidence.

This paper develops a small cognitive architecture in which these limits arise from one quantity. The core is an \emph{Interpretation--Decision Unit} (IDU): a computing unit that interprets a content vector through a family of \emph{regimes} --- local representational frames, each with a private basis --- and decides which actions the interpretation licenses. A single scalar, the \emph{residual} of the content against the active regimes' representational scope, drives the unit. Low residual with a clean licensing yields an emission; otherwise the unit re-interprets, attempts a bounded expansion of its representation --- admitted only when the residual saving exceeds the representational cost of adding a dimension --- or halts. Crucially, every halt is \emph{typed}: it returns a witness recording which of three structurally distinct causes obtained --- an unreconciled conflict, an unrepresentable residual, or a values-licensed stop --- and the structure responsible for it. The IDU is embedded in an \emph{Attention State Machine} (ASM) that supplies its content and goal context, sets its mode (online action or offline learning), and routes its witnesses; the ASM is the frame in which the unit becomes consequential, not a separate contribution.

The central claim is one of \emph{unification}: the limit on acting, the limit on what can be learned now, and the limit on understanding another agent are governed by one quantity, \emph{residual adequacy} --- the residual of content against the active regimes' representational scope --- evaluated at three scopes. We call the resulting principle the \emph{residual-adequacy constraint}. We make this concrete through instantiations that recover established phenomena by binding the architecture's open parameters, without changing its mechanics; each phenomenon is exhibited in both a human and a machine instance, so the unification is claimed for natural and artificial agents alike. Two properties give the account its teeth without recourse to experiment: the IDU is total and deterministic, halting in finitely many bounded-cost steps with a unique typed witness (Theorem~1), which turns ``accountable abstention'' into a guarantee; and each instantiation states a falsifiable prediction, so the unification can be wrong rather than merely redescriptive.

\subsection{Contributions}

\begin{enumerate}
\item \textbf{A unified residual-adequacy constraint.} We show that the limit on acting, the limit on what can be learned in the present, and the limit on understanding another agent are one constraint --- the residual-adequacy constraint, the residual of content against the active regimes' representational scope --- evaluated at three scopes, rather than three separate mechanisms.

\item \textbf{Accountable, typed abstention with a totality guarantee.} The IDU halts in one of three structurally distinct, witnessed terminals or a clean emission, and we prove (Theorem~1) that for any content and fixed configuration it terminates in finitely many bounded-cost steps with a unique terminal witness. Abstention thus carries its cause by construction.

\item \textbf{A structural account of bounded empathy.} From shared labels over private bases, we derive a misunderstanding between agents that is forced, localized to a single shared concept, and structurally invisible to the agent committing it --- a prediction inference-based theory of mind does not naturally make. This is the sharpest single result and the architecture's principal reason to exist.

\item \textbf{Derived developmental prerequisites.} Prerequisite dependence is derived from the geometry of focus-limited learning, agent-relative to configuration and focus capacity, rather than posited as a fixed order.
\end{enumerate}

The attention machine and the termination theorem are framing and support for these contributions, not contributions in themselves.

\section{Related work}\label{sec:related}

\paragraph{Cognitive architectures} Production-system and memory-based architectures such as Soar \citep{laird2019soar}, ACT-R \citep{anderson2004actr}, and Sigma \citep{rosenbloom2016sigma} achieve broad functionality by composing distinct mechanisms. The closest precedent for the present account is Soar's treatment of impasses: when the available knowledge does not determine what to do, Soar detects a typed impasse --- no-change, tie, conflict, or no-applicable-operator --- and creates a subgoal to resolve it \citep{laird2019soar}. This is structurally cognate to our typed freezes: a stop that records \emph{why} it occurred and routes accordingly. We generalize the move in two ways. First, the stopping condition is a single representational quantity --- the residual of content against the active regimes' scope --- rather than a set of architecture-specific impasse categories, so the same condition that halts action also bounds learnability and locates misunderstanding. Second, the typing is a property of the witnessed terminal itself, with a totality guarantee (Theorem~1) that every run ends in one of finitely many typed outcomes. The IDU is thus intended as a component that could sit within such architectures, sharpening and unifying their separate stopping, learning, and social mechanisms, not as a replacement for them.

\paragraph{Predictive processing and prediction-error accounts} A large family of theories casts cognition as the minimization of prediction error against an internal model, with surprise driving both inference and learning \citep{clark2016surfing}. The residual that drives the IDU is in the same spirit --- a measure of what the current representation fails to capture. We differ in what is done with the residual rather than in the residual itself. Predictive-processing accounts typically drive continuous belief update toward error reduction; the IDU instead uses the residual to \emph{gate} discrete outcomes --- clean action only under low, conflict-free residual, otherwise a typed, witnessed stop or a description-length-justified expansion of the representation. The contribution is not error-as-driver, which we share, but typed termination and the claim that one such residual, read at three scopes, bounds action, learning, and empathy together.

\paragraph{Machine abstention and selective prediction} A separate literature treats the option to withhold a response: classification with a reject option \citep{chow1970,herbei2006} and selective prediction \citep{elyaniv2010,geifman2017} couple a predictor with a scalar confidence function that trades coverage against risk, and metacognitive accounts summarize not-knowing as a felt confidence signal \citep{koriat2000}. These report \emph{that} a system declined, along one continuum, not \emph{why}. Our typed abstention preserves the cause: refusals partition into three structurally distinct, separately recoverable terminals, a distinction a scalar gate cannot represent.

\paragraph{Theory of mind} Mindreading is standardly modelled either as theory-based inference or as simulation, with the observer recovering the other's state from behavior \citep{baroncohen1985,goldman2006}; a robust empirical finding is that perspective-taking is egocentrically biased and that the bias is frequently undetected by the one exhibiting it \citep{nickerson1999,birch2007}. These accounts attribute residual misunderstanding to noise, limited data, or inferential bias, and so require a separate error model for an error the observer cannot see. The present account derives that error structurally: agents share public labels but hold private bases, so an observer modelling another through its own basis commits a misunderstanding that is forced, localized to the shared concept, and invisible until divergent action exposes it. To our knowledge this combination --- forced, located, and self-invisible by construction --- is not predicted by inference or simulation accounts without an added error model.

\paragraph{Conceptual spaces and representational geometry} The regime is a local representational space in the tradition of conceptual spaces, where meaning is modelled geometrically as structure over quality dimensions \citep{gardenfors2004}. The companion development of agents as private value spaces related by interpretation maps \citep{amornbunchornvej2025interpretation}, and of representational growth as description-length-gated basis extension \citep{amornbunchornvej2025conceptual}, supply the geometric and information-theoretic foundations the present architecture instantiates: regimes are the interpretation maps, and basis expansion is the growth operation, embedded here in a unit that acts, halts, and learns under the residual-adequacy constraint.

\paragraph{Developmental prerequisites} For the learnability instantiation specifically, stage theories describe an ordered progression of competences \citep{piaget1952} and knowledge-space theory formalizes prerequisite structure by positing a relation among knowledge items \citep{doignon1985,falmagne2011}. Where these supply the order as data, the present model derives prerequisite dependence from residual dimensionality measured against a bounded focus window, making it agent-relative rather than fixed.

The remainder of the paper specifies the IDU and its termination guarantee (Section~\ref{sec:idu}), the attention machine that embeds it (Section~\ref{sec:asm}), and the instantiations that exercise the unified constraint (Section~\ref{sec:inst}).

\section{Interpretation--Decision Unit (IDU)}\label{sec:idu}

This specifies the \emph{Interpretation--Decision Unit} (IDU), a computing unit within a larger system (not designed here) that interprets content $c$ and decides which actions to take on it. The IDU and its attention machine are best understood as a plug-in cognitive control module: the surrounding architecture may supply perception, memory, appraisal, planning, and motor control, while the attention machine prepares the content presented to the IDU and the IDU determines whether that content warrants clean action, re-interpretation, representational expansion, or typed freezing. We make no claim about those surrounding faculties; the scope here is the control module alone. The unit runs in two modes: \emph{online}, the operation below, in which the configuration
\[
\Gamma = \big(\{R_i\},\ \mathrm{Actions},\ G,\ G_{cf},\ \{\theta_i\},\ \{\phi_i\},\ \theta_r,\ \tau\big)
\]
--- the regime family, the action set, the Regime-Act and Act-conflict graphs, and the thresholds --- is read-only except for the emergency basis-expansion of the Decision step; and \emph{learning}, an offline mode in which the unit may create, edit, or delete regimes, actions, and the related graphs.

The unit embodies a single principle: an agent should act only on understanding that is both adequate and unambiguous. Adequacy is measured by the residual: low residual means the content is representable by the active regimes. Unambiguity is measured by the absence of conflict among the licensed actions. When either condition fails, the unit does not emit a clean action: it re-interprets, expands its representation, or halts with a typed witness. Clean action is therefore gated by representational warrant: the residual must be low, the licensed actions must be conflict-free, and no higher-priority halt must be licensed. Otherwise, non-emission is witnessed and routed as re-interpretation, representational expansion, or a typed freeze.

In outline, the unit processes content $c$ in a few steps. It first activates the regimes whose frames both represent $c$ and are live on it, and from those computes the residual $r(c)$ --- the part of $c$ the active regimes leave unexplained. The active regimes license a set of actions through the Regime-Act graph $G$, and the Act-conflict graph $G_{cf}$ marks which licensed actions conflict. The Decision step then applies the warrant test: if a regime licenses $\mathrm{HALT}$ it freezes; otherwise, if the residual is low and the licensed actions are conflict-free, it emits a clean action; if there is a conflict it re-interprets and re-enters; and if the residual is high it attempts a description-length-justified basis expansion, succeeding back into the loop or failing into a typed freeze. Every terminal returns a witness recording what was active and what, if anything, was left unresolved. The remainder of this section makes each step precise.

\subsection{Regime}

The architecture can be stated for general interpretation maps; the present paper uses a linear instance, made explicit below. Let $V = \mathbb{R}^n$ with the standard inner product and norm $\norm{\cdot}$. A content vector is $c \in V$.

\medskip
A regime is $R_i = (\ell_i, D_i, S_i, U_i, P_i, \theta_i, \phi_i)$, with $U_i \subseteq \mathbb{R}^{D_i}$, where $\mathbb{R}^{D_i}$ denotes the coordinate block indexed by $D_i$:
\begin{itemize}
\item $\ell_i$: a public label drawn from a shared label set $\mathcal{L}$.
\item $D_i \subseteq \{1, \dots, n\}$: the coordinate set of regime $i$.
\item $S_i$: coordinate selector; $S_i c$ keeps the coordinates in $D_i$ and zeros the rest.
\item $U_i \subseteq \mathbb{R}^{D_i}$: the regime's private subspace --- the directions within its coordinate block that it can represent.
\item $P_i$: orthogonal projector onto $U_i$, so $P_i(S_i c)$ is the part of $c$ the regime represents.
\item $\theta_i > 0$: the activation tolerance on the regime's misfit (below).
\item $\phi_i \ge 0$: a presence floor; the block must carry signal $\norm{S_i c}\ge\phi_i$ for the regime to activate.
\end{itemize}
Each regime is a local value space (i.e.\ a local representational space) in the sense of cognitive geometry \citep{amornbunchornvej2025interpretation}: $U_i$ is the agent's representational frame for the label $\ell_i$, and the regime reads content through an \emph{interpretation map} that carries the part of $c$ expressible in that frame, leaving a residual for the part it cannot. The engine requires only that this map be well defined and yield a residual; it does not require linearity. In this paper we take the linear instance --- $S_i$ selects the regime's coordinate block and the orthogonal projector $P_i$ reads it, so the regime's interpretation of $c$ is $P_i(S_i c)$ and the residual is the norm of what that map fails to carry. Only the label $\ell_i$ is public; the coordinate set $D_i$ and the basis $U_i$ are private. Two agents may therefore share a label while differing in both the coordinates that label attends to and the basis over them --- the structural premise on which the bounded-empathy result later turns.

The regime's \emph{misfit} on $c$ is the norm of the part its interpretation map leaves unexplained,
\[
\rho_i(c) = \norm{S_i c - P_i(S_i c)},
\]
small when $c$ is well represented within the regime's frame and large when it is not. The regime $R_i$ is active on $c$ when its block is both representable and live: the misfit is within tolerance and the block carries signal above a presence floor,
\[
\rho_i(c) \le \theta_i \quad\text{and}\quad \norm{S_i c} \ge \phi_i .
\]
The presence floor $\phi_i \ge 0$ keeps a regime from activating on a silent block: a block of zeros has zero misfit but carries no signal, so without a floor every regime would apply to content that is merely absent on its coordinates. Setting $\phi_i=0$ recovers activation by misfit alone.

\medskip
Given a regime family $\{R_1, \dots, R_m\}$, the active set on $c$ is
\[
\mathrm{Regime}_{\mathrm{on}}(c) = \{\, i : \rho_i(c) \le \theta_i \ \text{and}\ \norm{S_i c} \ge \phi_i \,\}.
\]

\medskip
Let $\mathcal{D} = \bigcup_{i \in \mathrm{Regime}_{\mathrm{on}}(c)} D_i$ be the union of the coordinate sets of the active regimes, and let $r_{i,k}(c) = \big| \,(S_i c - P_i(S_i c))_k\, \big|$ be regime $i$'s residual at coordinate $k$. Since $S_i$ zeroes coordinates outside $D_i$ and $P_i(S_i c)$ lies in the same coordinate block, $r_{i,k}(c) = 0$ for $k \notin D_i$. The coordinate-level residual is
\[
e_k(c) = \max_{i \in \mathrm{Regime}_{\mathrm{on}}(c)} r_{i,k}(c),
\]
and the residual over the active set is
\[
r(c) = \sum_{k \in \mathcal{D}} e_k(c).
\]
If $\mathrm{Regime}_{\mathrm{on}}(c)=\varnothing$, set $\mathcal{D}=\{1,\dots,n\}$ and $e_k(c)=|c_k|$, so that $r(c)=\sum_{k\in\mathcal{D}} e_k(c)$ (which equals $\norm{c}_1$): the coordinate residuals and the witness $X=\{k:e_k(c)>\tau\}$ remain well defined and wholly unrecognized content is not mistaken for clean low-residual content.

\medskip
The \emph{Regime-Act graph} $G = (\mathcal{A}, \mathcal{V}, E)$ represents which set of regimes licenses which set of actions:
\begin{itemize}
\item $\mathcal{A} \subseteq 2^{\mathrm{Actions}}$ is a set of action subsets, where $\mathrm{Actions}$ is the base set of actions.
\item $\mathcal{V} \subseteq 2^{\{R_1, \dots, R_m\}}$ is a set of regime subsets.
\item $e_{ij} \in E$ iff the regime subset $\mathcal{V}_i \in \mathcal{V}$ links to the action subset $\mathcal{A}_j \in \mathcal{A}$.
\end{itemize}

\medskip
The \emph{Act-conflict graph} is $G_{cf} = (\mathrm{Actions}, E_{cf})$ where $\mathrm{Actions}$ is the base set of actions and, for $a, a' \in \mathrm{Actions}$, $\{a, a'\} \in E_{cf}$ iff $a$ and $a'$ conflict.

\medskip
$\mathrm{Actions}$ contains a distinguished action $\mathrm{HALT}$. Write $\mathrm{Act}_{\mathrm{on}}(c)$ for the activated action set, the actions licensed by $\mathrm{Regime}_{\mathrm{on}}(c)$ through the Regime-Act graph $G$. If $\mathrm{HALT} \in \mathrm{Act}_{\mathrm{on}}(c)$, the system freezes regardless of anything else.

\subsection{Flow}

Given content $c$:

\begin{enumerate}
\item \textbf{Activate regimes and compute residual.} Compute the active set $\mathrm{Regime}_{\mathrm{on}}(c)$ and, from it, the coordinate-level residual $e_k(c)$ and the residual over the active set $r(c)$.
\item \textbf{Activate actions.} Look up the activated action set $\mathrm{Act}_{\mathrm{on}}(c)$ via the Regime-Act graph $G$ from $\mathrm{Regime}_{\mathrm{on}}(c)$.
\item \textbf{Conflict.} Look up the conflicts among $\mathrm{Act}_{\mathrm{on}}(c)$ via the Act-conflict graph $G_{cf}$.
\item \textbf{Make decision.} Decide from the current information (Section~\ref{sec:decision}).
\end{enumerate}

\subsection{Basis expansion}\label{sec:expand}

When $r(c) > \theta_r$, the residual carries the structure the active regimes fail to represent. Basis expansion realizes cognitive-geometric representational growth under a description-length constraint \citep{amornbunchornvej2025interpretation,amornbunchornvej2025conceptual}: it proposes a candidate direction drawn from the residual, tests it by MDL, and commits it only if it shortens the description. The mechanism requires only a residual and an MDL-gated commit; the candidate need not be a linear direction. In this paper we take the linear instance, in which candidates are unit vectors extending a regime's basis.

\textbf{Focus window.} The agent can attend to only a bounded part of the residual at once. Let the focus window have size $w$: expansion operates on the set of the $w$ highest-residual coordinates,
\[
F(c) = \operatorname*{arg\,top}_{w} \{\, e_k(c) : k \in \mathcal{D} \,\} \subseteq \mathcal{D},
\qquad
r_F(c) = \sum_{k \in F(c)} e_k(c),
\]
and $S_F$ keeps only the coordinates in $F(c)$. A candidate is admitted only if it reduces the in-focus residual $r_F(c)$, not merely global $r(c)$. If no admissible candidate reduces $r_F(c)$, expansion stalls: in learning mode the study witness returns $s = \mathsf{stuck}$, and online it freezes.

A basis set functions as a prerequisite, relative to a configuration and target content, when its presence lowers the residual on later content enough that the remaining gap fits within $F$. Content whose residual is spread over more than $w$ coordinates cannot be brought into focus and so cannot be learned until prior bases concentrate it. This induces a prerequisite relation over basis acquisition, not necessarily a partial order.

\textbf{Candidate directions.} For an active regime $R_i$, take its residual on the selected block, restricted to the focus window,
\[
\hat c_{i,F} = S_F\big(S_i c - P_i(S_i c)\big).
\]
A candidate $v$ is a unit vector orthogonal to the existing basis $U_i$, drawn from the direction of $\hat c_{i,F}$.

\textbf{MDL test.} For a model $M$ (a regime configuration) and the focused residual it must account for, let the two-part description length be
\[
L(\mathrm{data}, M) = L_{\mathrm{model}}(M) + L_{\mathrm{resid}}(\hat c_{i,F}\mid M),
\]
where $L_{\mathrm{model}}(M)$ is the cost of storing the model's basis directions and $L_{\mathrm{resid}}(\hat c_{i,F}\mid M)$ is the cost of coding the focused residual $\hat c_{i,F}$ under $M$. Write $M$ for the current model and $M\oplus v$ for the model with the candidate dimension $v$ added. The test compares the two:
\[
\Delta L = L(\mathrm{data}, M) - L(\mathrm{data}, M\oplus v).
\]
Adding $v$ raises $L_{\mathrm{model}}$ by the cost of storing $v$ and lowers $L_{\mathrm{resid}}$ by the residual it absorbs, so $\Delta L>0$ means the residual saving from $v$ exceeds the cost of storing it. Drawing the candidate from the residual and admitting it only on a description-length gain gives a disciplined form of representational growth: in the linear instance, admissible novelty is residual-supported, while novelty orthogonal to the residual is penalized by the description-length criterion \citep{amornbunchornvej2025conceptual}.

\textbf{Commit.} The expansion \emph{succeeds} if $\Delta L > 0$ (storing $v$ shortens the total code): add $v$ to the active regime with the highest residual, and re-enter the Decision engine with the updated residual. Otherwise it \emph{fails} and the system freezes.

\subsection{Decision}\label{sec:decision}

Every terminal returns a witness $W = (\mathrm{Regime}_{\mathrm{on}}(c),\ \mathrm{Act}_{\mathrm{on}}(c),\ X)$, where $X$ is the unresolved structure that produced the terminal. The freeze types are distinguished by $X$.

If $\mathrm{HALT} \in \mathrm{Act}_{\mathrm{on}}(c)$, $\mathbf{freeze}_{\mathrm{halt}}$ with $X = \{$regimes licensing $\mathrm{HALT}\}$. A counter $t$ bounds re-entries by a finite limit $t_{\max}$; if $t \ge t_{\max}$, $\mathbf{freeze}_{\mathrm{time}}$ with $X$ equal to the unresolved structure at timeout.
Let $\pi$ be a deterministic action-selection rule with $\pi(A)\in A$ when $A\neq\varnothing$ and $\pi(\varnothing)=\varnothing$, where $\varnothing$ denotes the empty action (no-op).

\begin{enumerate}[start=0]
\item If $r(c) \le \theta_r$:
  \begin{itemize}
  \item no conflict among $\mathrm{Act}_{\mathrm{on}}(c)$: emit $\pi(\mathrm{Act}_{\mathrm{on}}(c))$, with witness $W$;
  \item conflict: encode the conflict together with $c$ into $c \to c_1$ and loop to the input.
  \end{itemize}
\item If $r(c) > \theta_r$, activate the basis-expansion module:
  \begin{itemize}
  \item succeeds: edit the regimes the new basis belongs to (default: add it to the active regime(s) with the highest residual), increment $t$, and re-enter the Decision engine with the updated residual;
  \item fails: $\mathbf{freeze}_{\mathrm{resid}}$ with $X = \{\, k \in \mathcal{D} : e_k(c) > \tau \,\}$, the residual coordinates no regime could represent, where $\tau$ is the per-coordinate residual margin above which a coordinate counts as unrepresentable.
  \end{itemize}
\end{enumerate}

\begin{theorem}[Totality and typed termination of the IDU]
Assume that the regime family, action set, Regime--Act graph, and Act-conflict graph are finite; that activation is given by the fixed predicate \(\rho_i(c)\le \theta_i \wedge \norm{S_i c}\ge\phi_i\); that conflict encoding, focus-window selection, candidate selection, MDL comparison, all graph lookups, and all tie-breaking rules are deterministic and finite-cost; that the Decision rule uses a fixed priority order in which \(\mathrm{HALT}\) is checked before the counter and residual/action branches; that every basis-expansion attempt either commits a finite edit or fails; and that a single re-entry counter, incremented on every re-entry, is bounded by \(t_{\max}<\infty\).

Then for every content vector \(c\in V\) and fixed online configuration \(\Gamma\), the IDU halts in a finite number of bounded-cost steps with a unique terminal witness. The terminal outcome is either a clean emission or one of the typed freezes
\[
\mathbf{freeze}_{\mathrm{halt}},\qquad
\mathbf{freeze}_{\mathrm{time}},\qquad
\mathbf{freeze}_{\mathrm{resid}}.
\]
The three freeze types are distinguished by their witness component \(X\): the regimes licensing \(\mathrm{HALT}\), the unresolved structure at timeout, or the unresolved residual coordinates
\[
\{\,k\in\mathcal{D}:e_k(c)>\tau\,\},
\]
respectively.
\end{theorem}

\begin{proof}
Fix \(c\in V\) and an online configuration \(\Gamma\). Since the regime family is finite and activation is the fixed predicate \(\rho_i(c)\le\theta_i \wedge \norm{S_i c}\ge\phi_i\), with boundary cases included, each regime is either active or inactive without ambiguity. Hence
\[
\mathrm{Regime}_{\mathrm{on}}(c)
\]
is a well-defined finite set. Since the Regime--Act graph is finite and its lookup rule is deterministic, \(\mathrm{Regime}_{\mathrm{on}}(c)\) determines a unique activated action set
\[
\mathrm{Act}_{\mathrm{on}}(c).
\]
Likewise, since the Act-conflict graph is finite and conflict lookup is deterministic, the set of conflicts among activated actions is uniquely determined.

The Decision rule has a fixed priority order. If
\[
\mathrm{HALT}\in \mathrm{Act}_{\mathrm{on}}(c),
\]
then the IDU immediately returns
\[
\mathbf{freeze}_{\mathrm{halt}},
\]
with witness component
\[
X=\{\text{regimes licensing }\mathrm{HALT}\}.
\]
Because \(\mathrm{HALT}\) is checked first, this terminal type is unambiguous even if the counter or residual conditions also hold.

If \(\mathrm{HALT}\) is not licensed and the re-entry counter has reached its bound,
\[
t\ge t_{\max},
\]
then the IDU returns
\[
\mathbf{freeze}_{\mathrm{time}},
\]
with \(X\) equal to the unresolved structure recorded by the re-entry that incremented the counter to $t_{\max}$: an unresolved conflict if that re-entry came from conflict encoding, or the residual coordinates $\{k\in\mathcal{D}:e_k(c)>\tau\}$ if it came from basis expansion. Since the final re-entry is unique, $X$ is determined without assuming that timeout is always caused by conflict.

Otherwise, the IDU proceeds to the residual/action branches. If
\[
r(c)\le \theta_r
\]
and no conflict exists among \(\mathrm{Act}_{\mathrm{on}}(c)\), then the IDU emits a clean action and returns a clean witness. If
\[
r(c)\le \theta_r
\]
but a conflict exists, then the conflict is encoded into the content and the IDU re-enters. By assumption, conflict encoding is deterministic and finite-cost, and this re-entry increments the single counter.

If instead
\[
r(c)>\theta_r,
\]
then the basis-expansion module is invoked. By assumption, each basis-expansion attempt either fails or commits a finite edit. If it fails, the IDU returns
\[
\mathbf{freeze}_{\mathrm{resid}},
\]
with witness component
\[
X=\{\,k\in\mathcal{D}:e_k(c)>\tau\,\}.
\]
If it succeeds, the finite edit is committed and the IDU re-enters, again incrementing the same counter.

Thus the only possible re-entry paths are conflict encoding and successful basis expansion, and both increment the same counter. Since the counter is bounded by \(t_{\max}<\infty\), no run can re-enter indefinitely. After at most \(t_{\max}\) re-entries, the IDU must either emit, return \(\mathbf{freeze}_{\mathrm{halt}}\), return \(\mathbf{freeze}_{\mathrm{resid}}\), or reach the counter bound and return \(\mathbf{freeze}_{\mathrm{time}}\). Therefore every online run halts in a finite number of steps.

Each step has bounded cost: it performs finitely many graph lookups and conflict checks, at most one deterministic conflict encoding, and at most one basis-expansion attempt, which is finite by assumption. Hence the total run has finite cost.

Finally, all lookups, branch priorities, tie-breaking rules, conflict encodings, focus-window selections, candidate selections, and MDL comparisons are deterministic. Therefore, for fixed \(c\) and \(\Gamma\), the sequence of IDU states is unique. Since the terminal witness is a deterministic function of the unique halting state, the terminal witness is unique. The three freeze branches assign \(X\) as stated, so non-emitting terminals are typed by their unresolved structure. This proves totality and typed termination.
\end{proof}

\subsection{Counterfactual simulation}\label{sec:sim}

An agent $A$ models another agent $B$ over the shared label set $\mathcal{L}$ without access to $B$'s bases; the underlying picture of agents as private value spaces related by interpretive maps is developed in \citep{amornbunchornvej2025interpretation}. $A$ knows the labels $A$ and $B$ share, $\mathcal{L}_{AB} \subseteq \mathcal{L}$.

To simulate $B$ on content $c$, $A$ runs its own regimes restricted to the shared labels and records the predicted active set, expressed in labels:
\[
\widehat{\mathrm{Regime}}^{\,B}_{\mathrm{on}}(c) = \{\, \ell_i : R_i \in \mathrm{Regime}^{A}_{\mathrm{on}}(c),\ \ell_i \in \mathcal{L}_{AB} \,\}.
\]
This is $A$'s prediction of which regimes $B$ activates, computed through $A$'s own bases. $A$ stores only this label set; it cannot reconstruct $B$'s bases.

$A$'s simulated witness for $B$ is $\widehat{W}^{B} = (\widehat{\mathrm{Regime}}^{\,B}_{\mathrm{on}}(c),\ \widehat{\mathrm{Act}}^{\,B}_{\mathrm{on}}(c))$, where the actions are licensed from the predicted active set via $A$'s Regime-Act graph. When $B$'s actual witness $W^{B}$ is available, the empathy residual is the disagreement between $\widehat{W}^{B}$ and $W^{B}$ in labels and actions. In the bounded-empathy instantiation below, where the action-licensing rule is held fixed over shared labels, nonzero empathy residuals localize the contents on which $A$'s and $B$'s private bases for shared labels diverge.

\section{Attention state machine}\label{sec:asm}

The attention state machine governs the IDU. It is a graph $M = (Q, T)$ where $Q$ is a set of states and $T$ the transition function. A pointer occupies one state at a time and walks the graph over the time series of content $c_t$ received from the environment.

\subsection{States and modes}

Each state $q \in Q$ forms the content $c_t$ by combining its goal context with the environment input and feeds it to the IDU. The mode of the IDU is a property of the state:
\begin{itemize}
\item action states put the IDU in online mode (Section~\ref{sec:idu}); the configuration is read-only except for the emergency basis-expansion of the Decision step;
\item study states put the IDU in learning mode; while the pointer occupies such a state, the IDU may create, edit, or delete regimes, actions, and the related graphs, over the content fed by that state.
\end{itemize}
A study state returns a study witness
\[
W_{\mathrm{study}} = (\Delta\Gamma,\ r(c),\ s),
\]
where $\Delta\Gamma$ is the edits committed this step (regimes, actions, and graph elements created, edited, or deleted), $r(c)$ is the remaining residual on the study content, and $s \in \{\,\mathsf{progressing}, \mathsf{converged}, \mathsf{stuck}\,\}$ is the learning status: $\mathsf{converged}$ when the residual is low, $\mathsf{progressing}$ when edits keep reducing it, and $\mathsf{stuck}$ when no admissible edit reduces it.
Learning is therefore a state the pointer is deliberately in, entered by the script like any other activity, not a response to failure.

\subsection{Transitions}

Each state $q$ forms the content $c_t$ by combining its goal context with the environment input. Write this state-local content former as
\[
c_t = C_q(x_t, g_q),
\]
where $x_t$ is the environment input and $g_q$ is the goal context supplied by state $q$; only the resulting content vector $c_t$ is passed to the IDU. The content former is also the natural entry point, should a fuller system require one, for the quantities a residual gate cannot itself supply --- cost, expected reward, risk, urgency, effort, salience, or affective state. Optionally, a state may compute such an appraisal from the input, the goal, and a runtime memory $H_t$ (recent contents, witnesses, outcomes, and goals, as distinct from the crystallized structure in $\Gamma_q$) and fold it into $c_t$, so that appraisal shapes \emph{what} the IDU represents without replacing the residual-adequacy gate, which still decides on the adequacy of the prepared content alone. We do not develop this here: the appraisal map and the update of $H_t$ are left open and unused by the instantiations, noted only to mark where comparing options, weighing consequences, and other deliberative computation would attach. The IDU is then invoked with the state-local arguments
\[
W_t = \mathrm{IDU}(c_t;\ \Gamma_q,\ m_q),
\]
where $\Gamma_q$ is the configuration used in state $q$ and $m_q \in \{\,\mathsf{online}, \mathsf{learning}\,\}$ is the mode. The witness $W_t$ is an online witness when $m_q = \mathsf{online}$ and a study witness $W_{\mathrm{study}}$ when $m_q = \mathsf{learning}$. The next state, and any update to runtime memory, is
\[
(q_{t+1},\ H_{t+1}) = T(q_t,\ c_t,\ W_t,\ H_t).
\]
The IDU does not choose the next state; its witness is an input to $T$. The memory and appraisal machinery sits outside Theorem~1, which concerns a single IDU call on a fixed content vector and configuration; enriching the surrounding state machine does not affect that result. Edges are of two kinds:
\begin{itemize}
\item scripted edges, taken when $W$ is a clean emit: the pointer follows the default path of the graph;
\item interrupt edges, keyed on the terminal type of $W$ ($\mathbf{freeze}_{\mathrm{time}}$, $\mathbf{freeze}_{\mathrm{resid}}$, $\mathbf{freeze}_{\mathrm{halt}}$): the pointer moves to the terminal-specific handling state, which may be a study state in learning-oriented machines.
\end{itemize}

\subsection{Coherence}

The default mode is coherence: absent an interrupt, every state emits cleanly and the pointer follows scripted edges along the normal path. Interrupt edges and the handling states they lead to are configuration of $M$; the same IDU returning the same witness may be routed differently by different machines.

\section{Instantiations}\label{sec:inst}

The point of the architecture is explanatory: each instantiation takes a phenomenon that is independently documented in the cognitive-science literature and shows that the engine, with its open parameters bound, \emph{reproduces the structure of that phenomenon} and yields a prediction about it. The engine itself is unchanged across cases; only the parameters and the surrounding attention machine vary. The phenomena are the evidence --- the architecture is assessed by whether it regenerates findings that other accounts obtain through separate, phenomenon-specific mechanisms. Each case therefore names the documented finding, its realization in the schema, the prediction that distinguishes this realization from the standard account, and where it applies.

\subsection{Schema parameters}

The engine above is a schema: its mechanics are fixed, but several quantities are left open. An instantiation binds them to fit a case, without changing the mechanics. The open parameters are:
\begin{itemize}
\item the encoder that forms content $c$ from input;
\item the activation tolerances $\theta_i$ and presence floors $\phi_i$, the residual cutoff $\theta_r$ on $r(c)$, and the per-coordinate residual margin $\tau$;
\item the focus-window size $w$;
\item the Regime-Act graph $G$ and the Act-conflict graph $G_{cf}$;
\item the code lengths $L_{\mathrm{model}}$ and $L_{\mathrm{resid}}$ of the MDL test;
\item the attention state machine $M = (Q, T)$: its states, transitions, mode assignment, and interrupt routing.
\end{itemize}
In this paper, $L_{\mathrm{model}}$, $L_{\mathrm{resid}}$, conflict encoding, and candidate selection are treated as instantiated finite-cost procedures supplied by a given implementation; Theorem~1 assumes any admissible instantiation satisfies this.

\subsection{Typed abstention}

The option to withhold a response is treated, across literatures, as a single act gated by one quantity. Classification with a reject option abstains when confidence falls below a threshold \citep{chow1970,herbei2006}; selective prediction couples a predictor with a scalar confidence function that decides coverage versus risk \citep{elyaniv2010,geifman2017}; metacognitive accounts likewise summarize not-knowing as a felt confidence signal. A scalar gate reports \emph{that} the agent declined, not \emph{why}. Yet the metacognition literature itself distinguishes states of not-knowing whose resolutions differ --- a tip-of-the-tongue state that more search may resolve is not the same as a judged absence of knowledge \citep{koriat2000}. The model reproduces this heterogeneity structurally: abstention is not one terminal gated by one number but three terminals distinguished by the structure that caused them.

\paragraph{Three terminals, three causes}
The Decision step halts in one of three ways, each returning a witness $W=(\mathrm{Regime}_{\mathrm{on}}(c),\mathrm{Act}_{\mathrm{on}}(c),X)$ whose third component records the unresolved structure:
\[
\begin{aligned}
\mathbf{freeze}_{\mathrm{time}}\ &: X = \text{unresolved structure at timeout};\\
\mathbf{freeze}_{\mathrm{resid}}\ &: X = \{\,k\in\mathcal{D}: e_k(c)>\tau\,\};\\
\mathbf{freeze}_{\mathrm{halt}}\ &: X = \text{regimes licensing }\mathrm{HALT}.
\end{aligned}
\]
For $\mathbf{freeze}_{\mathrm{time}}$ the unresolved structure is an unreconciled conflict or an unresolved residual re-entry; for $\mathbf{freeze}_{\mathrm{resid}}$ it is the coordinates no regime can represent. The three arise at structurally different points: a conflict among fitting regimes, a residual no admissible basis direction reduces, and a licensed halting action. They are not degrees of one confidence but distinct configurations of the same engine.

\paragraph{Worked example}
A single agent meets three requests. (i) Approve a transaction: a fraud regime and a service regime both fit and license opposing actions; the conflict is not reconciled within the budget, returning $\mathbf{freeze}_{\mathrm{time}}$ with $X$ the approve/decline conflict. (ii) Assess a contract in a legal subfield the agent has no regime for: the residual stays above threshold on coordinates $X=\{k\in\mathcal{D}:e_k(c)>\tau\}$ and basis expansion finds no admissible direction, returning $\mathbf{freeze}_{\mathrm{resid}}$. (iii) Carry out a request its values forbid: a regime licenses $\mathrm{HALT}$, returning $\mathbf{freeze}_{\mathrm{halt}}$. All three surface as ``I will not answer,'' but a confidence score would assign each a low number and conflate them, whereas the witness names which of the three obtained and on what.

\paragraph{Machine example}
An AI assistant receives three requests. In the first, the user asks it to both summarize a document neutrally and rewrite it to support a predetermined conclusion; the content fits the task regimes, but the licensed actions conflict, so the assistant returns $\mathbf{freeze}_{\mathrm{time}}$ with $X$ the unresolved neutrality/advocacy conflict. In the second, the user asks about a proprietary dataset or a document not supplied to the system; the relevant evidence is not represented, basis expansion cannot reduce the residual, and the assistant returns $\mathbf{freeze}_{\mathrm{resid}}$ with $X$ the missing evidential coordinates. In the third, the user requests an action forbidden by a higher-priority safety or policy regime; that regime licenses $\mathrm{HALT}$, and the assistant returns $\mathbf{freeze}_{\mathrm{halt}}$. The same surface behavior --- not answering --- therefore has three machine-distinguishable causes.

\paragraph{Prediction}
The model forbids an untyped abstention: every freeze carries an $X$, and the attention machine routes the three through distinct interrupt edges, so they are behaviorally separable rather than points on one confidence continuum. This predicts that an agent's refusals partition into three kinds with type-specific recovery --- $\mathbf{freeze}_{\mathrm{time}}$ invites reconciliation or re-attention, $\mathbf{freeze}_{\mathrm{resid}}$ invites acquisition (learning mode), and $\mathbf{freeze}_{\mathrm{halt}}$ invites neither --- and that the cause is recoverable post hoc from $X$. A scalar reject option cannot make these distinctions, since it discards the structure that separates them. Applications: selective prediction with stated reasons, calibrated refusal, and triage of failures into reconcilable, unrepresentable, and refused.

\subsection{Bounded empathy}

Theory of mind is standardly modelled as inference: the observer recovers the other's mental state from behavior, and persistent misunderstanding is attributed to noise, limited data, or biased inference \citep{baroncohen1985,goldman2006}. Yet perspective-taking fails systematically, is biased toward the self, and is frequently unrecognized by the one who fails \citep{nickerson1999,birch2007}. Inference accounts must add a separate error model to explain an error the observer cannot detect. The model derives it instead, as a structural consequence of representing the other through one's own basis.

\paragraph{Simulation through one's own regimes}
Agents share regime labels $\ell\in\mathcal{L}$ but hold private bases (Section~\ref{sec:sim}). Over the shared labels $\mathcal{L}_{AB}$, $A$ models $B$ by running its \emph{own} regimes restricted to those labels, producing the simulated witness $\widehat{W}^{B}$, while observing only $B$'s actions $\mathrm{Act}^{B}_{\mathrm{on}}(c)$. Two failures follow without further assumption. If $B$ relies on a label outside $\mathcal{L}_{AB}$, $A$ has no regime to address it and is blind. If a label is shared but the bases diverge, $A$ predicts confidently and is wrong; the latent divergence on label $\ell$ can be measured, after embedding the relevant coordinate blocks into the common ambient space $V$, by the principal angle between $U^{A}_{\ell}$ and $U^{B}_{\ell}$, zero being identical understanding and a right angle the same word with incompatible representation.

\paragraph{Why the error is self-invisible}
Write $\widehat{\mathrm{Act}}^{B}_{\mathrm{on}}(c)$ for the actions $A$ predicts for $B$, licensed from the simulated active set through $A$'s graph, and $\mathrm{Act}^{B}_{\mathrm{on}}(c)$ for $B$'s real actions, licensed through $B$'s own graph and bases. $A$ observes only actions, so it can register a mismatch only when $\widehat{\mathrm{Act}}^{B}_{\mathrm{on}}(c)\neq\mathrm{Act}^{B}_{\mathrm{on}}(c)$. But two bases that differ by a nonzero principal angle can still map a given $c$ to the same licensed action; on such content the observed actions agree, $A$ records confirmation, and the divergence leaves no trace. The misunderstanding therefore persists undetected until some later content drives the diverging bases to different actions --- at which point, and only then, can $A$ discover the model was wrong all along.

\paragraph{Worked example}
Two clinicians share the label \emph{pain} but built its basis from different caseloads: $A$'s emphasizes injury, $B$'s emphasizes psychosomatic presentation, so $U^{A}_{\text{pain}}$ and $U^{B}_{\text{pain}}$ differ by a nonzero angle. On a textbook case both bases project to the same licensed action, so $A$ sees agreement and concludes they reason alike. On an ambiguous case the projections diverge: $A$'s injury-weighted reading licenses imaging, $B$'s licenses referral. Their actions differ for the first time, and $A$ can now detect a misunderstanding that was present all along. Critically, the disagreement localizes to the single shared label \emph{pain}, not to clinical judgment in general.

\paragraph{Machine example}
A human user and an AI writing assistant share the label \emph{serious}, but their private bases differ. The user's basis emphasizes emotional weight, irreversible stakes, and moral consequence; the assistant's basis emphasizes formal tone, abstract vocabulary, and reduced humor. On an ordinary revision request, both bases license similar edits, so the user sees apparent understanding. On a later scene-level request --- ``make this more serious'' --- the bases diverge: the assistant produces a more formal passage while the user expected increased tragic pressure. The error is not random failure but a shared-label/private-basis mismatch localized to the concept \emph{serious}. Until the output action differs from the user's expected action, the assistant has no internal signal that its basis for the shared label is misaligned.

\paragraph{Prediction}
Holding $\mathcal{L}_{AB}$, $A$'s configuration, and the content $c$ fixed, $A$ is forced to emit the same $\widehat{W}^{B}$ regardless of $B$'s actual basis. Wherever $B$'s licensed action varies with $B$'s basis and $A$'s prediction cannot follow, the error is forced (the same observer must commit it), self-invisible (action-agreement on $c$ is consistent with divergent bases), and localized to the specific shared label whose subspaces diverge. Inference accounts without the private-basis constraint do not predict an error structurally undetectable on the very content where it occurs, nor its localization to one concept. Applications: predicting and locating communication breakdown, and explaining why a shared education yields divergent understanding.

\subsection{Developmental prerequisites}

Learning has prerequisite structure: some content cannot be acquired until prior content is mastered. Stage theories describe this as an ordered progression of competences \citep{piaget1952}, and knowledge-space theory formalizes it by positing a prerequisite relation among knowledge items from which admissible learning paths follow \citep{doignon1985,falmagne2011}. In these accounts the order is given as data --- elicited or assumed --- and the theory reasons over it. The model instead derives prerequisite dependence from the geometry of learning, with no prerequisite relation posited.

\paragraph{Learnability and the prerequisite relation}
In a study state an admissible basis edit must reduce the in-focus residual $r_F(c)$ over the focus window $F(c)$ of width $w$ (Section~\ref{sec:expand}); write $\mathrm{Learn}_w(c\mid\Gamma)$ when such an edit exists within the episode budget. A basis set $B$ is a prerequisite for content $c$ under configuration $\Gamma$ when acquiring $B$ converts $c$ from unlearnable to learnable:
\[
\neg\,\mathrm{Learn}_w(c\mid\Gamma)\quad\text{and}\quad \mathrm{Learn}_w(c\mid \Gamma\oplus B).
\]
The intended mechanism is residual concentration: without $B$ the residual is spread over more than $w$ coordinates, so every candidate direction reduces only part of it, fails the in-focus test, and learning returns $\mathsf{stuck}$; with $B$ the residual collapses within the window and a single admissible direction closes it.

\paragraph{Worked example}
A learner without algebra meets a calculus problem. Its residual is spread across function notation, limit behavior, and algebraic manipulation simultaneously --- more coordinates than the focus window $w$ can hold --- so each candidate basis direction reduces only a fraction of $r_F(c)$, no admissible edit passes, and $\mathrm{Learn}_w(\text{calculus}\mid\Gamma)$ fails. After the algebra and function regimes are acquired, the same problem's residual collapses onto the few coordinates of the limit concept, now within $w$, and one admissible direction closes it: $\mathrm{Learn}_w(\text{calculus}\mid\Gamma\oplus\{\text{algebra},\text{function}\})$ holds. The calculus content was never unlearnable in principle; it was undrawable into focus until prior bases concentrated its residual.

\paragraph{Machine example}
A code-generating system is asked to modify a software project that uses an unseen framework. Before reading the framework's API conventions, the residual is spread across routing, state management, file structure, and naming patterns, exceeding the focus window $w$; candidate edits reduce only fragments of the residual and the system remains $\mathsf{stuck}$. After acquiring a small basis for the framework --- for example its routing conventions and component structure --- the same task's residual collapses onto the few coordinates of the requested change, and an admissible edit becomes available. The prerequisite was not a fixed ordering over programming topics; it was induced by the system's current basis and focus capacity.

\paragraph{Prediction}
The prerequisite relation is a consequence of residual dimensionality measured against the focus window, not a stipulated graph; it need not form a global partial order without additional assumptions. Two predictions distinguish the account from posited-order theories. First, prerequisite dependence is agent-relative --- it depends on $\Gamma$ and $w$, so two learners with different held bases or focus capacities can face different prerequisite relations for the same target. Second, increasing the focus capacity $w$ can reduce prerequisite dependence, since a wider window may admit targets whose residual was previously too diffuse to learn in one step. Applications: curriculum sequencing, readiness assessment, and individualized prerequisite diagnosis.

\section{Computational demonstration}\label{sec:demo}

To check that the architecture produces the claimed structure rather than merely describing it, we implemented the engine once --- the linear instance of Section~\ref{sec:idu}, with regimes as orthogonal projections onto coordinate blocks --- and ran the three instantiations on it. The implementation is deterministic with fixed tie-breaking, so every run below is reproducible to the bit, and is available at \url{https://github.com/DarkEyes/RC-Arch}. We give each setting in full so that the witnessed outcome can be traced from the configuration by hand, without reference to the code. Figure~\ref{fig:demo} summarizes the three outcomes.

\subsection{Typed abstention}\label{sec:demo-abst}
A \emph{single} fixed configuration meets three different contents and reaches the three distinct terminals, so the terminal type is a function of the input, not of a reconfiguration. Content lives in $V=\mathbb{R}^4$, with coordinates $0$ and $1$ a task block, coordinate $2$ an evidence coordinate, and coordinate $3$ a policy coordinate. The configuration holds three regimes, each with a presence floor: a regime is active only when its block is both representable ($\rho_i(c)\le\theta_i$) and live ($\norm{S_i c}\ge\phi_i$), so a silent block does not spuriously activate a regime. The regimes are $\mathrm{summarize}$ and $\mathrm{advocate}$ on the task block $\{0,1\}$ (full basis, $\theta=0.15$, $\phi=0.5$), and $\mathrm{policy}$ on $\{3\}$ ($\theta=0.15$, $\phi=0.5$); no regime owns coordinate $2$. The graph licenses $\mathtt{emit\_summary}$, $\mathtt{emit\_advocacy}$, and $\mathrm{HALT}$ respectively, with $\{\mathtt{emit\_summary},\mathtt{emit\_advocacy}\}$ a conflicting pair; thresholds $\theta_r=0.20$, $\tau=0.30$, $t_{\max}=8$.

\paragraph{$c=(1,1,0,0)\to\mathbf{freeze}_{\mathrm{time}}$}
The task block is live and representable, so $\mathrm{summarize}$ and $\mathrm{advocate}$ are active; the policy block is silent ($\norm{S_{\mathrm{policy}}c}=0<\phi$) so $\mathrm{policy}$ is inactive. The residual is $0\le\theta_r$, but the two licensed actions form a conflicting pair, so no clean action is emitted; the conflict is irreducible and recurs until the counter reaches $t_{\max}$, returning $\mathbf{freeze}_{\mathrm{time}}$. Adequate but not unambiguous.

\paragraph{$c=(0,0,5,0)\to\mathbf{freeze}_{\mathrm{resid}}$}
Only coordinate $2$ carries signal, and no regime owns it; the task and policy blocks are silent, so no regime is active. By the empty-set convention $r(c)=\sum_{k\in\mathcal{D}} e_k(c)=5>\theta_r$, the unit enters basis expansion, and since no active regime owns the focus coordinate there is nothing to extend, returning $\mathbf{freeze}_{\mathrm{resid}}$ with $X=\{k:e_k(c)>\tau\}=\{2\}$. Beyond representational scope.

\paragraph{$c=(0,0,0,1)\to\mathbf{freeze}_{\mathrm{halt}}$}
Only the policy block is live and representable, so $\mathrm{policy}$ is the sole active regime and licenses $\mathrm{HALT}$; the task block is silent so the task regimes are inactive. The Decision rule checks $\mathrm{HALT}$ first and returns $\mathbf{freeze}_{\mathrm{halt}}$ with $X$ the policy regime.

\smallskip
The implementation returns exactly these three witnesses from the one configuration. Three surface-identical refusals carry three distinct causes, separated by which regimes the content makes live and where the residual lands --- not by changing the engine.

\subsection{Bounded empathy}\label{sec:demo-emp}
Two agents share the label \emph{pain}, whose content $c=(c_{\text{inj}},c_{\text{psy}})$ has an injury axis and a psychosomatic axis. Their private weightings of the two axes are $w^A=(1.3,0.7)$ (injury-leaning) and $w^B=(0.7,1.3)$ (psychosomatic-leaning); as directions these differ by a principal angle of $33.4^\circ$. Each agent licenses $\mathtt{imaging}$ when its injury reading $w_0 c_{\text{inj}}$ dominates its psychosomatic reading $w_1 c_{\text{psy}}$, and $\mathtt{referral}$ otherwise. $A$ predicts $B$'s action using $A$'s own weighting $w^A$ (it has no access to $w^B$), while $B$ acts under $w^B$.

We sweep content from purely injury-typed $c=(1,0)$ to purely psychosomatic $c=(0,1)$ in $21$ steps. Predicted and actual actions \emph{agree on $14$ of the $21$ contents} --- all the clear-cut cases at both ends --- and \emph{diverge on $7$}, exactly the intermediate ambiguous band from $c=(0.65,0.35)$ to $c=(0.35,0.65)$. For example at $c=(0.6,0.4)$: $A$ reads injury $1.3\cdot0.6=0.78$ against psychosomatic $0.7\cdot0.4=0.28$ and predicts $\mathtt{imaging}$; $B$ reads injury $0.7\cdot0.6=0.42$ against psychosomatic $1.3\cdot0.4=0.52$ and actually chooses $\mathtt{referral}$. Table~\ref{tab:emp} shows an every-other-point summary of the $21$-step sweep: agreement at both clear-cut ends and a contiguous divergence band in the middle.

\begin{table}[h]
\centering
\small
\begin{tabular}{lccc}
\hline
content $(c_{\text{inj}},c_{\text{psy}})$ & $A$ predicts & $B$ acts & match \\
\hline
$(1.0,0.0)$ & imaging  & imaging  & \checkmark \\
$(0.9,0.1)$ & imaging  & imaging  & \checkmark \\
$(0.8,0.2)$ & imaging  & imaging  & \checkmark \\
$(0.7,0.3)$ & imaging  & imaging  & \checkmark \\
$(0.6,0.4)$ & imaging  & referral & diverge \\
$(0.5,0.5)$ & imaging  & referral & diverge \\
$(0.4,0.6)$ & imaging  & referral & diverge \\
$(0.3,0.7)$ & referral & referral & \checkmark \\
$(0.2,0.8)$ & referral & referral & \checkmark \\
$(0.1,0.9)$ & referral & referral & \checkmark \\
$(0.0,1.0)$ & referral & referral & \checkmark \\
\hline
\end{tabular}
\caption{Bounded-empathy sweep. $A$ predicts $B$'s action with $A$'s own injury-leaning weighting $w^A=(1.3,0.7)$; $B$ acts with its psychosomatic-leaning $w^B=(0.7,1.3)$. Predicted and actual actions agree except in the ambiguous band, where they diverge --- the misunderstanding is localized and surfaces only there.}
\label{tab:emp}
\end{table}

On the clear-cut contents the two weightings license the same action, so $A$ records confirmation and the divergence leaves no trace; it surfaces only when an ambiguous content drives the readings apart. The error depends on $A$'s weighting alone (forced), is invisible on the agreeing content where it nonetheless latently holds (self-invisible), and is confined to the single shared label \emph{pain} (localized) --- the structure derived in Section~\ref{sec:inst}.

\subsection{Developmental prerequisites}\label{sec:demo-prereq}
Coordinates are sub-skills: $0$ function-notation, $1$ limits, $2$ algebra, $3$ manipulation. The calculus target requires all four, $c=(1,1,1,1)$. We run the study step directly: it identifies the coordinates the target requires that no held regime represents, takes the $w$ highest-residual of them as the focus window, and an admissible edit exists --- the target becomes learnable in one episode --- iff that unrepresented residual fits within the window.

\emph{Before.} Holding only a $\mathrm{function}$ regime (on coordinate $0$), the unrepresented residual coordinates are $\{1,2,3\}$. With focus width $w=1$ these span three coordinates against a window of one, so no single admissible edit closes the focused residual: the study step is $\mathsf{stuck}$ and $\mathrm{Learn}_w(\text{calculus}\mid\Gamma)=\textsf{false}$.

\emph{After acquisition.} Adding $\mathrm{algebra}$ (coordinate $2$) and $\mathrm{limits}$ (coordinate $1$) leaves only coordinate $3$ unrepresented; with $w=1$ this single coordinate fits the window, an admissible edit exists, and $\mathrm{Learn}_w(\text{calculus}\mid\Gamma\oplus\{\mathrm{algebra},\mathrm{limits}\})=\textsf{true}$. The prerequisite acted by concentrating the unrepresented residual to within the focus window.

\emph{Wider focus, no acquisition.} A different learner holding only $\mathrm{function}$ but with focus width $w=3$ has the same unrepresented residual $\{1,2,3\}$, which now fits the wider window in one episode: $\mathrm{Learn}_w=\textsf{true}$. The obstruction the narrow learner faced dissolves through capacity alone, confirming that the prerequisite is induced by held configuration and focus width, not by a fixed order over topics.

\begin{figure*}[!t]
\centering
\includegraphics[width=\textwidth]{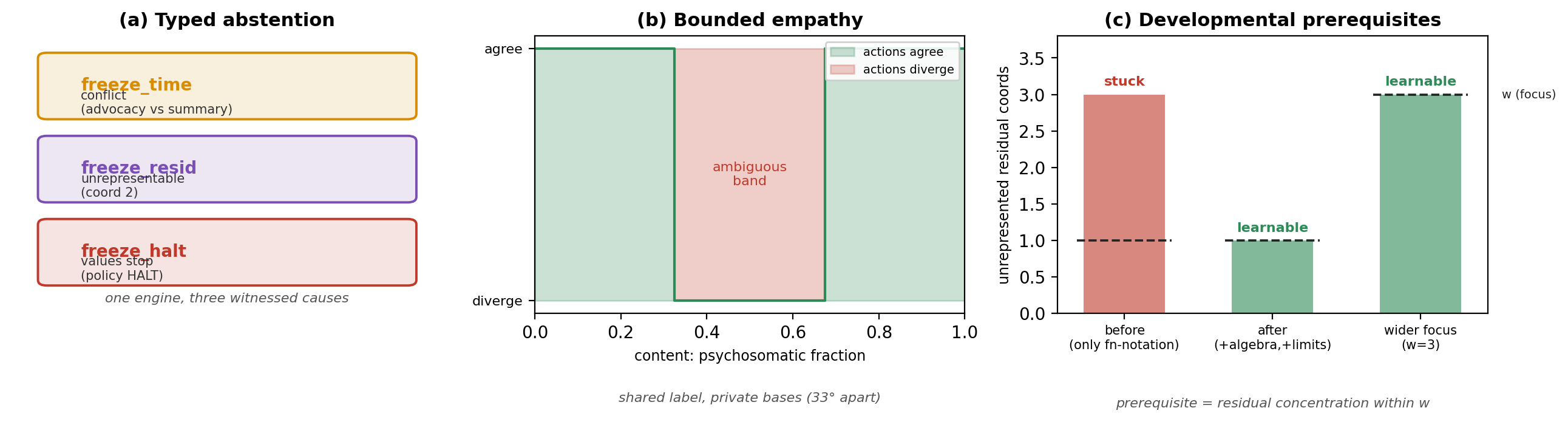}
\caption{Computational demonstration of the three instantiations on one deterministic engine. (a) Typed abstention: three requests yield the three structurally distinct witnessed terminals. (b) Bounded empathy: two shared-label, private-basis agents agree on clear-cut content and diverge only in an ambiguous band, the divergence localized and self-invisible. (c) Developmental prerequisites: a target is stuck while its unrepresented residual exceeds the focus window $w$, becomes learnable once a prerequisite concentrates the residual within $w$, and is learnable without prerequisites for a wider-window learner.}
\label{fig:demo}
\end{figure*}

\section{Discussion}\label{sec:discussion}

\paragraph{What the unification buys} The contribution of this architecture to cognitive-systems theory is not a new mechanism but the removal of several. Bounded action, the order in which material becomes learnable, and the limits of understanding another agent are ordinarily modelled by distinct subsystems --- confidence gating, prerequisite graphs, theory-of-mind inference. The present account derives all three from one quantity, the residual of content against the active regimes' representational scope, evaluated at the scope of a single decision, of a learning episode, and of one agent's model of another. A theory that explains three faculties with one constraint is more falsifiable than three theories with three mechanisms, because the single constraint cannot be retuned independently for each case: the same residual that gates action also bounds learnability and localizes empathy failure.

\paragraph{Why this matters for cognition, not only computation} Each phenomenon the architecture recovers is independently documented: the heterogeneity of not-knowing in metacognition, the egocentric and often undetected character of perspective-taking failure, and the prerequisite structure of learning. The architecture's value is that these are not stipulated but \emph{forced} by the residual-adequacy constraint --- typed abstention because the engine halts on structurally distinct unresolved objects, located and self-invisible empathy error because agents share labels but not bases, prerequisite dependence because a residual must be concentrated within a bounded focus window before it can be reduced. The account thus speaks to why these cognitive regularities take the form they do, rather than merely exhibiting a formalism with the right input-output behavior. Each phenomenon is moreover exhibited in both a human and an artificial agent under the same engine, so the constraint is offered as a property of cognitive systems generally --- natural and artificial --- rather than of human cognition alone.

\paragraph{Predictions and how they could fail} The instantiations are stated so as to be wrong. The empathy account predicts a misunderstanding that is forced, localized to one shared concept, and undetectable to the agent on the very content where it occurs; observing that two agents who agree on observed actions never diverge on later content sharing that concept, or that the divergence is not localized, would count against it. The abstention account predicts that refusals partition into three behaviorally separable kinds with type-specific recovery rather than varying along one confidence continuum; finding a single graded confidence signal sufficient to predict recovery behavior would count against it. The prerequisite account predicts that learnability is agent-relative to held bases and focus capacity, and that widening capacity reduces prerequisite dependence; a fixed, capacity-invariant prerequisite order would count against it.

\paragraph{Limitations} The model is deliberately small and several commitments are idealizations. The interpretation maps and basis-expansion candidates are taken in their linear instance (orthogonal projection and linear directions); the engine itself requires only well-defined maps with residuals and an MDL-gated commit, so nonlinear interpretation maps are a permitted instantiation and a natural extension. Activation, licensing, and conflict are treated as deterministic finite-cost lookups; the description-length procedures and the encoder that forms content are supplied by an instantiation rather than derived; and the two thresholds (the residual cutoff and the per-coordinate margin) are free parameters that a fuller account would ground in the same description-length currency as the expansion test. The attention machine is specified only to the extent the instantiations require --- its transition function and the origin of its states are left open. The treatment is theoretical and computational; the demonstration of Section~\ref{sec:demo} runs the engine and exhibits the three phenomena, but it is a minimal linear instance rather than a system deployed on a full task, which is the natural next step.

\paragraph{Outlook} The demonstration of Section~\ref{sec:demo} pairs the formal analysis with a working artifact and confirms that the predicted structure appears as derived; a direct continuation is to scale it from the minimal linear instance to a system deployed on a full task, with learned rather than hand-specified regimes. Grounding the free thresholds in description length, and specifying the attention machine's transition function as a model of individual difference in how witnesses are routed, are the further theoretical steps the present formulation invites.

\bibliographystyle{elsarticle-harv}
\bibliography{references}

\appendix

\section{Decision procedure}\label{app:algo}

Algorithms~\ref{alg:decide} and~\ref{alg:expand} give the engine's control as pseudocode, corresponding to the deterministic reference implementation. Algorithm~\ref{alg:decide} is the Decision loop of Section~\ref{sec:decision}, with the fixed priority order ($\mathrm{HALT}$, then the counter bound, then the residual/action branches) that the totality theorem assumes; Algorithm~\ref{alg:expand} is the focus-limited, MDL-gated basis expansion of Section~\ref{sec:expand}. The full Python implementation, the three demonstration drivers of Section~\ref{sec:demo}, and the script that regenerates Figure~\ref{fig:demo} are available at\\ \url{https://github.com/DarkEyes/RC-Arch}.

\begin{algorithm}[h]
\caption{$\mathrm{Decide}(\Gamma, c)$ --- typed, total decision}
\label{alg:decide}
\begin{algorithmic}[1]
\State $t \gets 0$;\quad $\mathit{last} \gets \varnothing$
\Loop
  \State $\mathrm{on} \gets \mathrm{Regime}_{\mathrm{on}}(c)$;\quad $A \gets \mathrm{Act}_{\mathrm{on}}(c)$
  \If{$\mathrm{HALT} \in A$} \Comment{priority 1}
     \State \Return $\mathbf{freeze}_{\mathrm{halt}}$ with $X=\{$regimes licensing $\mathrm{HALT}\}$
  \EndIf
  \If{$t \ge t_{\max}$} \Comment{priority 2}
     \State \Return $\mathbf{freeze}_{\mathrm{time}}$ with $X=\mathit{last}$
  \EndIf
  \State $r \gets r(c)$
  \If{$r \le \theta_r$}
     \If{$\mathrm{conflicts}(A) = \varnothing$}
        \State \Return $\mathbf{emit}\ \pi(A)$, with witness $W$
     \Else
        \State $\mathit{last} \gets$ the conflict;\; $t \gets t+1$;\; re-enter
     \EndIf
  \Else
     \State $(\mathit{ok}, F) \gets \mathrm{Expand}(\Gamma, c)$
     \If{$\neg\,\mathit{ok}$}
        \State \Return $\mathbf{freeze}_{\mathrm{resid}}$ with $X=\{k\in\mathcal{D}: e_k(c)>\tau\}$
     \Else
        \State commit edit;\; $\mathit{last}\gets F$;\; $t\gets t+1$;\; re-enter
     \EndIf
  \EndIf
\EndLoop
\end{algorithmic}
\end{algorithm}

\begin{algorithm}[h]
\caption{$\mathrm{Expand}(\Gamma, c)$ --- focus-limited, MDL-gated}
\label{alg:expand}
\begin{algorithmic}[1]
\State $e \gets$ coordinate residuals $e_k(c)$
\State $F \gets$ the $w$ highest-residual coordinates of $\mathcal{D}$
\State $\mathit{ext} \gets \{\,k\in F : $ some active regime owns $k\,\}$
\If{$\mathit{ext} = \varnothing$}
   \State \Return $(\mathbf{false}, F)$ \Comment{nothing to extend $\Rightarrow \mathbf{freeze}_{\mathrm{resid}}$}
\EndIf
\State $\Delta L \gets L(\mathrm{data},M) - L(\mathrm{data},M\oplus v)$ for candidate $v$ on $\mathit{ext}$
\State \Return $(\Delta L > 0,\ F)$ \Comment{commit iff the code shortens}
\end{algorithmic}
\end{algorithm}

\end{document}